\documentclass[twocolumn,showpacs,superscriptaddress,pra]{revtex4}%
\pdfoutput=1
\usepackage{amsmath}
\usepackage{amsfonts}
\usepackage{graphicx}
\usepackage{longtable}
\usepackage{amssymb}%
\providecommand{\U}[1]{\protect\rule{.1in}{.1in}}
\newtheorem{theorem}{Theorem}

\newcommand{\qed}{{\hfill$\Box$}}
\newenvironment{proof}{\noindent \textbf{{Proof~} }}{\qed}

\let\originalleft\left
\let\originalright\right
\def\left#1{\mathopen{}\originalleft#1}
\def\right#1{\originalright#1\mathclose{}}

\begin{document}

\title{Quantum discord and classical correlation can tighten the uncertainty 
principle in the presence of quantum memory}
 
\author{Arun Kumar Pati}
\affiliation{Harish-Chandra Research Institute, Chhatnag Road, Jhunsi, 
Allahabad 211 019, India,}
\author{Mark M. Wilde}
\affiliation{School of Computer Science, McGill University, Montreal,
 Quebec H3A 2A7, Canada,}
\author{A. R. Usha Devi}
\affiliation{Department of Physics, Bangalore University, Bangalore-560 056, 
India,}
\affiliation{Inspire Institute Inc., Alexandria, Virginia, 22303, USA,}
\author{A. K. Rajagopal}
\affiliation{Harish-Chandra Research Institute, Chhatnag Road, Jhunsi, 
Allahabad 211 019, India,}
\affiliation{Inspire Institute Inc., Alexandria, Virginia, 22303, USA,}
\author{Sudha}
\affiliation{Inspire Institute Inc., Alexandria, Virginia, 22303, USA,}
\affiliation{Department of Physics, Kuvempu University, Shankaraghatta, 
Shimoga-577 451, India.}
\date{\today}

\begin{abstract}
Uncertainty relations capture the essence of the inevitable randomness
associated with the outcomes of two incompatible quantum measurements.
Recently, Berta \textit{et al}.\ have shown that the lower bound on the
uncertainties of the measurement outcomes depends on the correlations between
the observed system and an observer who possesses a quantum memory. If the
system is maximally entangled with its memory, the outcomes of two
incompatible measurements made on the system can be predicted precisely. Here,
we obtain a new uncertainty relation that tightens the lower bound of Berta
\textit{et al}., by incorporating an additional term that depends on the
quantum discord and the classical correlations of the joint state of the
observed system and the quantum memory. We discuss several examples of states
for which our new lower bound is tighter than the bound of Berta \textit{et
al}. On the application side, we discuss the relevance of our new inequality
for the security of quantum key distribution and show that it can be used to
provide bounds on the distillable common randomness and the entanglement of
formation of bipartite quantum states.

\end{abstract}

\maketitle

The Heisenberg uncertainty relation is one of the fundamental concepts in
quantum theory \cite{Heisenberg}.  In recent
years, entropic uncertainty relations (EURs)~\cite{MU} have gained importance
because of their operational applications for privacy issues in cryptographic
tasks \cite{Berta}, in the detection of entangled states \cite{Berta}, and in
constructing error-correcting codes for communicating quantum or private
classical information \cite{RB08,WR12,WR12a}. A revision of the Maasen-Uffink
EUR \cite{MU}, which includes the possibility of a quantum
memory correlated with the observed system, has been identified recently by
Berta \textit{et al}.~\cite{Berta}. Interestingly, if the system is maximally
entangled with its memory, the outcomes of two incompatible measurements made
on distinct and identical copies of such a state can
be predicted precisely by an observer with access to the
quantum memory. It may be noted that the possibility of violating the 
uncertainty relation using an entangled pair was conceived as early as 
1934 by Karl Popper~\cite{Popper}, who had proposed an experiment to do so.

The uncertainty principle sets limits on our ability to predict the outcomes
of two incompatible measurements, and it was originally formulated by
Heisenberg~\cite{Heisenberg} for position and momentum observables.
Robertson~\cite{rob} and Schr\"{o}dinger~\cite{sch} generalized it to
arbitrary pairs of non-commuting observables $P$ and $Q$ and it is well known
in the following form:%
\begin{equation}
(\Delta P)(\Delta Q)\geq\frac{1}{2}\,|\langle\lbrack P,Q]\rangle|,
\end{equation}
where the uncertainties in the measurements are quantified in terms of the
standard deviations $(\Delta P)\equiv\sqrt{\langle P^{2}\rangle-\langle
P\rangle^{2}}$, $(\Delta Q)\equiv\sqrt{\langle Q^{2}\rangle-\langle
Q\rangle^{2}}$, and the commutator $[P,Q]\equiv PQ-QP$.

Deutsch subsequently advocated for characterizing the \textquotedblleft
spread\textquotedblright\ in the outcomes of two incompatible measurements via
Shannon entropies%
\begin{equation}
H(P)\equiv-\sum_{i}p_{i}\log_{2}p_{i},\ \ \ H(Q)\equiv-\sum_{j}q_{j}\log
_{2}q_{j}, \label{eq:shannon-ent}%
\end{equation}
where $p_{i}$ and $q_{j}$ are the probability distributions of the 
measurement outcomes,
rather than with standard deviations \cite{Deutch} because EURs provide an
information-theoretic basis for quantifying uncertainties \cite{Wehner}. The
entropic uncertainty bound for position and momentum observables was first
proposed by Hirschmann~\cite{Hirschmann} more than five decades ago, and it
was improved further during later years~\cite{Beckner, bial1}. The entropic
constraints for the outcomes of an arbitrary pair of observables $P$ and $Q$
were formulated by Deutsch~\cite{Deutch} and were subsequently improved by
various authors \cite{partovi, bial2, kraus}. A conjecture put forth by
Kraus~\cite{kraus} was proved by Maassen and Uffink~\cite{MU}, leading to the
following EUR:%
\begin{equation}
H(P)+H(Q)\geq-2\log_{2}c(P,Q), \label{mubound}%
\end{equation}
where $c(P,Q)\equiv\mathrm{max}_{i,j}|\langle p_{i}|q_{j}\rangle|$,
$\{|p_{i}\rangle\}$, $\{|q_{j}\rangle\}$ are the eigenvectors of $P$ and $Q$,
respectively. Later, this bound was improved and tightened by several authors
\cite{KKRP,Berta,Coles11,Coles} to hold for arbitrary POVMs and to include a
state-dependent term in the lower bound depending on the entropy of the
state of the system:%
\begin{equation}
H\left(  P\right)  +H(Q)\geq-2\log_{2}c\left(  P,Q\right)  +S\left(
\rho\right)  , \label{berta-improvement}%
\end{equation}
where $P$ is now more generally a positive operator-valued measure
(POVM)\ with elements $\left\{  \Lambda_{p}\right\}  $, $Q$ is a POVM\ with
elements $\left\{  \Gamma_{q}\right\}  $, the entropies $H\left(  P\right)  $
and $H\left(  Q\right)  $ are the Shannon entropies of the
distributions $\text{Tr}\left\{  \Lambda_{p}\rho\right\}  $ and $\text{Tr}%
\left\{  \Gamma_{q}\rho\right\}  $, respectively, with $\rho$ the state of the
quantum system, $c\left(  P,Q\right)  $ is now an incompatibility measure 
for the POVMs:
$
c\left(  P,Q\right)  \equiv\max_{p,q}\left\Vert \sqrt{\Lambda_{p}}\sqrt
{\Gamma_{q}}\right\Vert _{\infty}^{2},
$
and the entropy $S\left(  \rho\right)  $ is the von Neumann entropy of $\rho$:
$S\left(  \rho\right)  \equiv-$Tr$\left\{  \rho\log\rho\right\}  $.

The inequality in (\ref{berta-improvement}) bounds the entropies of the outcomes of
$P$ and $Q$ in
the system when it is not correlated with a quantum memory. Accessibility of a
quantum memory (so that the system $A$ and the memory $B$ are in some state
$\rho_{AB}$) leads to a refinement of
(\ref{berta-improvement})~\cite{Berta}:
\begin{equation}
S(P|B)+S(Q|B)\geq-2\log_{2}c(P,Q)+S(A|B), \label{bertabound}%
\end{equation}
where the lower bound is expressed similarly to that in
(\ref{berta-improvement}): as a sum of the measurement incompatibility
$-2\log_{2}c(P,Q)$ and the state-dependent term $S(A|B)$. The
quantities\ $S(P|B)\equiv S(\rho_{PB})-S(\rho_{B})$ and $S(Q|B)\equiv
S(\rho_{QB})-S(\rho_{B})$ are the conditional von Neumann entropies of the
post measurement states on the measurement register and the memory:
\begin{align}
\rho_{PB}  &  =\sum_{p}\left\vert p\right\rangle \left\langle p\right\vert
_{P}\otimes\text{Tr}_{A}\left\{  (\Lambda_{p}\otimes I)\rho_{AB}\right\}
,\label{eq:p-post-measure-state}\\
\rho_{QB}  &  =\sum_{q}\left\vert q\right\rangle \left\langle q\right\vert
_{Q}\otimes\text{Tr}_{A}\{(\Gamma_{q}\otimes I)\rho_{AB}\},
\label{eq:q-post-measure-state}%
\end{align}
respectively resulting after the POVMs $P$ and $Q$ are performed on the system
$A$, with additional information stored in the memory $B$. The entropies
$S(P|B)$ and $S(Q|B)$ quantify the uncertainty about the outcome of a
measurement on the system from the perspective of a party holding the quantum
memory $B$. Here, $S(A|B)\equiv S(\rho_{AB})-S(\rho_{B})$ is the conditional
von Neumann entropy between $A$ and $B$. The additional term $S(A|B)$ on the
right hand side can become negative when the system $A$ is entangled with its
quantum memory $B$, suggesting that it should be possible to reduce the sum of
the conditional uncertainties down to zero. Indeed, suppose that the
measurements $P$ and $Q$ are two canonical-conjugate observables and that the
system and memory are in a maximally entangled state. In this case, the sum
$S(P|B)+S(Q|B)$ is equal to zero, meaning that the possessor of the quantum
memory can predict the outcome of either measurement on the system simply by
measuring his quantum memory. The RHS\ of the above inequality is consistent
with this---when the system and memory are in a maximally entangled state, the
conditional entropy $S(A|B)$ assumes its most negative value so that
$S(A|B)=-\log_{2}\,d$ (where $d$ denotes the dimension of the system) and as
$-2\log_{2}c(P,Q)$ cannot exceed the value $\log_{2}d$~\cite{Wehner}, the
RHS of (\ref{bertabound}) reduces to zero. This EUR\ has been
recently experimentally tested \cite{exptl1,exptl2}.

One of the major goals in quantum information theory is to understand and
exploit various resources like entanglement, quantum correlations
\cite{zurek}, and classical correlations \cite{ved} in a composite quantum
state. Entanglement has been used extensively in many quantum information
processing tasks. However, the role of quantum and classical correlations is
beginning to be understood and is a subject of extensive research in recent
years \cite{kavan}. Our objective here is to understand how other correlation
measures such as quantum discord \cite{zurek}\ and classical correlation
\cite{ved} can play a role in tightening the EUR\ in (\ref{bertabound}).
This leads to two new EURs, and we will show that the first inequality
provides a new lower bound on the regularized entanglement of 
formation for any density operator shared between Alice and Bob. Also, 
we show that it can be used to
give an upper bound on the distillable common randomness
\cite{devetak,kw} $C_{D}^{\rightarrow}(\rho_{CB})$ with a third-party system.
The second inequality can be used to provide an improved bound 
for the secure key rate in a quantum key distribution protocol, in the case that
Alice and Bob have some description of the state that they share.

In this paper, we prove the following new entropic uncertainty relation:

\begin{theorem}
\label{main-theorem}The uncertainties $S(P|B)$ and $S(Q|B)$ are lower bounded
by the measurement incompatibility $-2\log_{2}c(P,Q)$, the conditional entropy
$S(A|B)$ of the state $\rho_{AB}$, and the larger of zero and the difference
between the state's quantum discord $D_{A}\left(  \rho_{AB}\right)  $ and its
classical correlation $J_{A}(\rho_{AB})$:
\begin{multline}
S(P|B)+S(Q|B)\geq-2\log_{2}c(P,Q)+S\left(  A|B\right) \label{arunbound}\\
+\max\left\{  0,D_{A}\left(  \rho_{AB}\right)  -J_{A}(\rho_{AB})\right\}  .
\end{multline}

\end{theorem}

The classical correlation $J_{A}(\rho_{AB})$ is defined as \cite{ved}%
\[
J_{A}(\rho_{AB})\equiv\max_{\left\{  \Upsilon_{x}\right\}  }I\left(
X;B\right)  ,
\]
where the mutual information $I\left(  X;B\right)  \equiv S\left(  X\right)
+S\left(  B\right)  -S\left(  XB\right)  $ is with respect to the
post-measurement state%
\[
\rho_{XB}=\sum_{x}\left\vert x\right\rangle \left\langle x\right\vert
_{X}\otimes\text{Tr}_{A}\left\{  \left(  \Upsilon_{x}\otimes I\right)
\rho_{AB}\right\}  ,
\]
and the optimization is over all POVMs $\left\{  \Upsilon_{x}\right\}  $
acting on the system $A$---note that it suffices to consider rank-one POVMs in
this optimization \cite{kavan}, because for every POVM\ that is not rank-one,
one can construct from it a rank-one POVM\ that gives a higher classical
correlation. If a state has no quantum correlations, the classical correlation
$J_{A}(\rho_{AB})$ is equal to the mutual information $I\left(  A;B\right)
\equiv S\left(  \rho_{A}\right)  +S\left(  \rho_{B}\right)  -S\left(
\rho_{AB}\right)  $ itself~\cite{ved}. The notion of quantum discord and
classical correlation between two systems plays an important role in the study of
quantum correlations. If one defines the total correlation using mutual
information, then the difference between the total and classical correlations
gives the so-called quantum discord~\cite{zurek}:%
\begin{equation}
D_{A}\left(  \rho_{AB}\right)  \equiv I\left(  A;B\right)  -J_{A}(\rho_{AB}),
\label{def:discord}%
\end{equation}
and it is a measure of quantum correlations in the state~$\rho_{AB}$.

We now prove the uncertainty relation as stated in (\ref{arunbound}).

\begin{proof}
[Theorem~\ref{main-theorem}] First consider that it suffices to prove the
following inequality:%
\begin{multline}
S(P|B)+S(Q|B)\geq-2\log_{2}c(P,Q)+S\left(  A|B\right) \label{eq:new-ineq}\\
+D_{A}\left(  \rho_{AB}\right)  -J_{A}(\rho_{AB}),
\end{multline}
because combining it with the bound in (\ref{bertabound}) leads to the bound
in (\ref{arunbound}). So, consider a bipartite density operator $\rho_{AB}$.
If Alice performs the POVM $P$, then the post-measurement state is as in
(\ref{eq:p-post-measure-state}), and similarly it is as in
(\ref{eq:q-post-measure-state}) if she chooses to perform the POVM\ $Q$. We
then have%
\begin{align}
&  S\left(  P|B\right)  +S\left(  Q|B\right) \nonumber\\
&  =H\left(  P\right)  -I\left(  P;B\right)  +H\left(  Q\right)  -I\left(
Q;B\right) \nonumber\\
&  \geq H\left(  P\right)  +H\left(  Q\right)  -2J_{A}\left(  \rho_{AB}\right)
\nonumber\\
&  \geq-2\log_{2}c\left(  P,Q\right)  +S\left(  A\right)  -2J_{A}\left(
\rho_{AB}\right) \nonumber\\
&  =-2\log_{2}c\left(  P,Q\right)  +S\left(  A|B\right)  +D_{A}\left(
\rho_{AB}\right)  -J_{A}(\rho_{AB}). \label{proof-of-new-ineq}%
\end{align}
The first equality is an identity for the quantum mutual information. The
first inequality follows from the definition of the classical correlation
$J_{A}\left(  \rho_{AB}\right)  $ and the fact that the POVMs $P$ or $Q$ may
not necessarily be the maximizing POVM\ for $J_{A}\left(  \rho_{AB}\right)  $,
so that $I\left(  P;B\right)  \leq J_{A}\left(  \rho_{AB}\right)  $ and
$I\left(  Q;B\right)  \leq J_{A}\left(  \rho_{AB}\right)  $. In the second
inequality, we apply the uncertainty relation in (\ref{berta-improvement}).
The second equality exploits the identity $S\left(  A\right)  =S\left(
A|B\right)  +I\left(  A;B\right)  $ and the
definition of the quantum discord in (\ref{def:discord}).
\end{proof}

Our new lower bound in (\ref{arunbound}) tightens the state-dependent term of
the bound in (\ref{bertabound}) if the discord $D_{A}\left(  \rho_{AB}\right)
$\ is larger than the classical correlation $J_{A}\left(  \rho_{AB}\right)  $.
This occurs for several natural examples of bipartite states, including Werner
states and isotropic states (see the appendices).

We discuss some examples to illustrate (\ref{arunbound}). Our first two
examples recover the result of Berta \textit{et al}. When the system is
uncorrelated with the memory, $\rho_{AB}=\rho_{A}\otimes\rho_{B}$, we have
$S(P|B)=H(P)$, $S(Q|B)=H(Q)$, $S(A|B)=S(A)$, and
 $D_{A}(\rho_{AB})=J_{A}(\rho_{AB})=0$,
resulting in the relation in (\ref{berta-improvement}). As a second example,
we consider when the system and memory are in a pure maximally entangled
state. In this case, $J_{A}(\rho_{AB})=S(\rho_{B})=\log_{2}d$, $I\left(
A;B\right)  =2\log_{2}d$, $D_{A}\left(  \rho_{AB}\right)  =\log_{2}d$,
$-2\,\log_{2}c(P,Q)\leq\log_{2}d$ and hence the right hand side of
(\ref{arunbound}) reduces to zero---as in the Berta \textit{et al}%
.~inequality---which is consistent with the fact that a quantum memory which
is maximally entangled with the system assists in predicting the measurement
outcomes of non-commuting observables precisely. In fact, our bound reduces
to the bound of Berta {\it et al}.~for all pure bipartite states, since the discord is equal to the
classical correlation for these states. This reduction also occurs
whenever the state is ``classical-quantum'' with $A$ classical and $B$ quantum.
The discord is equal to zero for these states \cite{kavan}, so that it is always
less than the classical correlation.

Our next example illustrates an important difference between our new
inequality in (\ref{arunbound}) and the Berta \textit{et al}.~inequality in
(\ref{bertabound}). We consider a two-qubit Werner state:%
\begin{equation}
\rho_{AB}=\frac{1-p}{4}I_{A}\otimes I_{B}+p\,|\Phi^{-}\rangle_{AB}\langle
\Phi^{-}|, \label{wer}%
\end{equation}
where $|\Phi^{-}\rangle_{AB}=(|0_{A},1_{B}\rangle-|1_{A},0_{B}\rangle
)/\sqrt{2}$ is the Bell state and $0\leq p\leq1$. Measurements of $\sigma_{x}$
and $\sigma_{z}$ lead to the conditioned entropies $S(\sigma_{x}%
|B)=S(\sigma_{z}|B)=h\left(  \left(  1-p\right)  /2\right)  ,$ where $h(x)$ is the binary entropy, so that%
\[
S(\sigma_{x}|B)+S(\sigma_{z}|B)=2h\left(  \left(  1-p\right)  /2\right)  .
\]
The incompatibility $c(\sigma_{x},\sigma_{z})=1/\sqrt{2}$, the classical
correlation in the state $J_{A}(\rho_{AB})=1-h\left(  \left(  1-p\right)
/2\right)  $, the entropy $S\left(  A\right)  =1$, and the conditional entropy
$S(A|B)=-\frac{3(1-p)}{4}\log_{2}\,\frac{(1-p)}{4}-\frac{1+3p}{4}\log
_{2}\,\frac{1+3p}{4}$. Recall from (\ref{proof-of-new-ineq}) that our lower
bound is equivalent to $-2\log_{2}c\left(  P,Q\right)  +S\left(  A\right)
-2J_{A}\left(  \rho_{AB}\right)  $, which for this example becomes
$1+1-2\left[  1-h\left(  \left(  1-p\right)  /2\right)  \right]  =2h\left(
\left(  1-p\right)  /2\right)  $. Thus the entropy sum $S(\sigma
_{x}|B)+S(\sigma_{z}|B)$ \textit{coincides with our lower bound} for all
values of $p$, a significant tightening of the Berta \textit{et al}%
.~inequality. In Fig.~1 we have plotted the right hand
sides of the inequalities in (\ref{bertabound}) and (\ref{arunbound}) to compare
the bounds.

\begin{figure}[ptb]
\includegraphics*[width=3.5in,keepaspectratio]{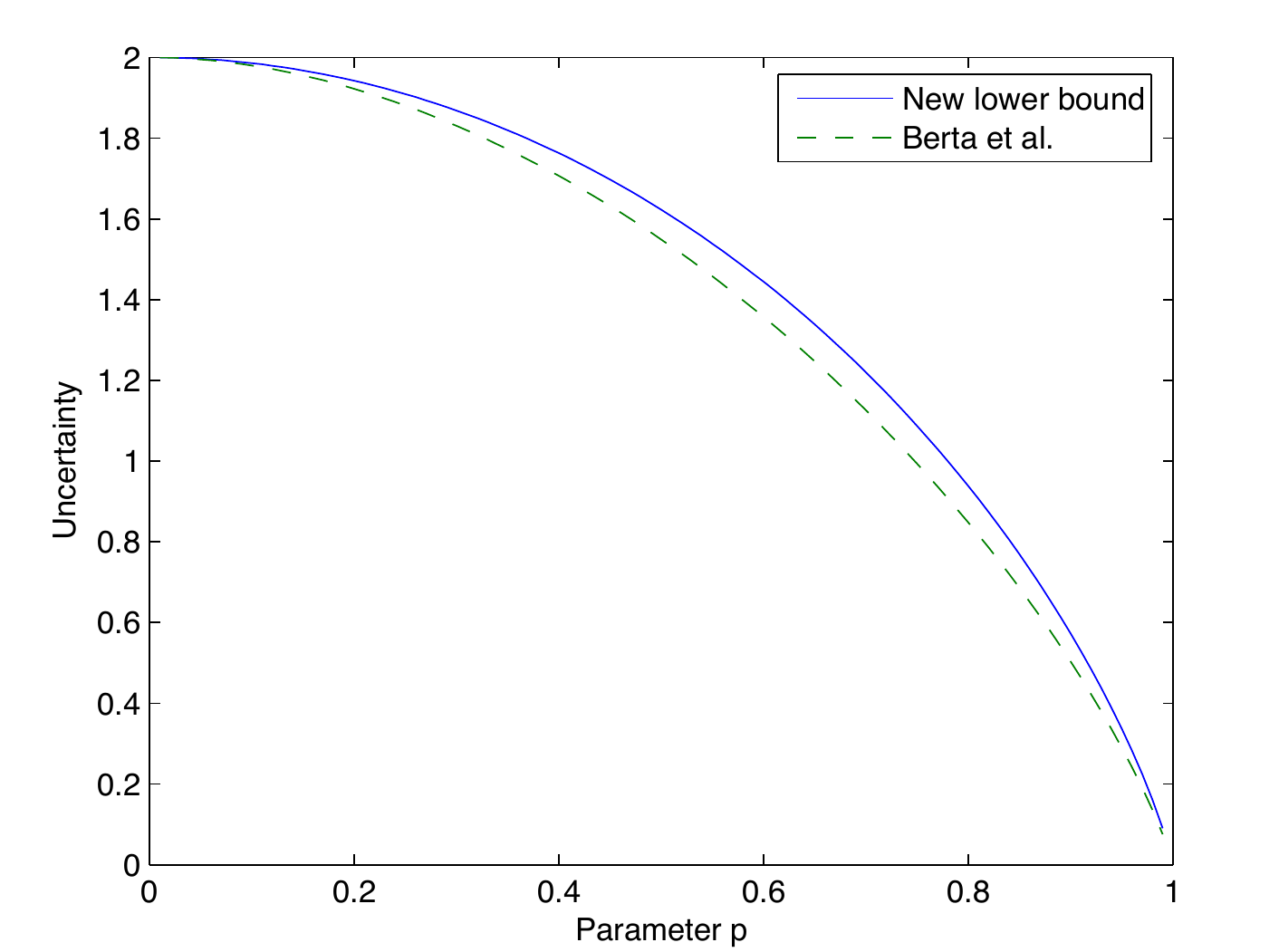}%
\caption{The right hand side of our new inequality (\ref{arunbound}) $-2
\log_{2}\, c(\sigma_{x},\sigma_{z}) + S(A|B) + \mathrm{max}\{0,D_{A}(\rho
_{AB})-\, J_{A}(\rho_{AB})\}$ (solid blue), and the right hand side of the
Berta \textit{et al}.~inequality (\ref{bertabound}) $-2\log_{2}\, c(\sigma
_{x},\sigma_{z})+S(A|B)$ (dashed green), as a function of the noise parameter
$p$, when the system and memory are prepared in the two qubit Werner state
(\ref{wer}).}%
\end{figure}

Interestingly, what occurs in the above Werner state example extends to
higher-dimensional Werner and isotropic states. When the measurement is of
Fourier-conjugate observables, the uncertainty sum $S(P|B)+S(Q|B)$ coincides
with our lower bound in (\ref{arunbound}) for all dimensions and parameters of
these states, demonstrating that our lower bound is perfectly tight for these
states and measurements. Hence, these states may be called {\it minimum 
entropic uncertain states} of the EUR in (\ref{arunbound}). We refer the reader to the appendices for 
details of this result, where we exploit recent findings of Chitambar 
\cite{chitambar12}. Note that minimum uncertainty states of the Berta
{\it et al}.~bound were studied in Refs.~\cite{PhysRevLett.103.020402,Coles}.

Berta \textit{et al}.~proved the following EUR for a tripartite state
$\rho_{ABE}$ \cite{Berta}:
$
S(P|B)+S(Q|E)\geq-2\log_{2}c(P,Q),
$
which has implications for the security of quantum key distribution (QKD). We
can further improve the lower bound to be as follows:
\begin{multline}
S(P|B)+S(Q|E)\geq-2\log_{2}c(P,Q)\label{eq:crypto-ineq}\\
+\max\left\{  0,D_{A}(\rho_{A|\left(  BE^{\prime}\right)  })-J_{A}\left(
\rho_{AB}\right)  \right\}  ,
\end{multline}
where the discord $D_{A}(\rho_{A|\left(  BE^{\prime}\right)  })$ is between
the $A$ and $BE^{\prime}$ systems of a purification $\left\vert \psi
\right\rangle _{ABEE^{\prime}}$\ of the state $\rho_{ABE}$ (note that
$D_{A}(\rho_{A|\left(  BE^{\prime}\right)  })=D_{A}(\rho_{AB})$ if the state
$\rho_{ABE}$ is already pure). The above EUR implies tighter security bounds
on the ability of an eavesdropper $E$ to predict the outcome of the $Q$
measurement on the system $A$, again whenever the discord~$D_{A}%
(\rho_{A|\left(  BE^{\prime}\right)  })$ is larger than the classical
correlation~$J_{A}(\rho_{AB})$ and in the case where Alice and Bob have some
description of the state that they share. It might not always be possible to
have such a description of the state, but in the case that they do, the above inequality
may lead to improved bounds on the secure key rate.

The inequality in (\ref{eq:crypto-ineq}) follows by considering a purification
$\left\vert \phi\right\rangle _{ABEE^{\prime}}$ of the state $\rho_{ABE}$ and
the following chain of inequalities that follow from reasoning similar to that
in (\ref{proof-of-new-ineq}):%
\begin{align}
&  S\left(  P|B\right)  +S\left(  Q|E\right) \label{eq:crypto-EUR-proof}\\
&  =H\left(  P\right)  -I\left(  P;B\right)  +H\left(  Q\right)  -I\left(
Q;E\right) \nonumber\\
&  \geq H\left(  P\right)  +H\left(  Q\right)  -J_{A}\left(  \rho_{AB}\right)
-J_{A}\left(  \rho_{AE}\right) \nonumber\\
&  \geq-2\log_{2}\left(  c\right)  +S\left(  A\right)  -J_{A}\left(  \rho
_{AB}\right)  -J_{A}\left(  \rho_{AE}\right) \nonumber
\end{align}
We now focus on rewriting the term $S\left(  A\right)  -J_{A}\left(  \rho
_{AE}\right)  $:%
\begin{align*}
S\left(  A\right)  -J_{A}\left(  \rho_{AE}\right)   &  =S\left(  A\right)
-S\left(  E\right)  +\min_{\left\{  \Lambda_{a}\right\}  }S\left(  E|\left\{
\Lambda_{a}\right\}  \right) \\
&  =S\left(  A\right)  -S\left(  ABE^{\prime}\right)  +\min_{\left\{
\Lambda_{a}\right\}  }S\left(  BE^{\prime}|\left\{  \Lambda_{a}\right\}
\right) \\
&  =S\left(  A\right)  +S\left(  BE^{\prime}\right)  -S\left(  ABE^{\prime
}\right) \\
&  \ \ \ \ \ -\left[  S\left(  BE^{\prime}\right)  -\min_{\left\{  \Lambda
_{a}\right\}  }S\left(  BE^{\prime}|\left\{  \Lambda_{a}\right\}  \right)
\right] \\
&  =I\left(  A;BE^{\prime}\right)  -J_{A}(\rho_{A|\left(  BE^{\prime}\right)
})\\
&  =D_{A}(\rho_{A|\left(  BE^{\prime}\right)  })
\end{align*}
The first equality follows by definition. The second equality follows by
noting that the state on systems $ABEE^{\prime}$ is pure, so that $S\left(
E\right)  =S\left(  ABE^{\prime}\right)  $. Also, the optimal POVM\ performed
on $A$ is a rank-one POVM, so that the conditional state on $EBE^{\prime}$ is
pure, implying that $\min_{\left\{  \Lambda_{a}\right\}  }S\left(  E|\left\{
\Lambda_{a}\right\}  \right)  =\min_{\left\{  \Lambda_{a}\right\}  }S\left(
BE^{\prime}|\left\{  \Lambda_{a}\right\}  \right)  $. The other equalities
follow from definitions. Substituting back in (\ref{eq:crypto-EUR-proof}), we
obtain%
\[
S\left(  P|B\right)  +S\left(  Q|E\right)  \geq-2\log_{2}\left(  c\right)
+D_{A}(\rho_{A|\left(  BE^{\prime}\right)  })-J_{A}\left(  \rho_{AB}\right)
\]
(Again, it suffices to prove the inequality above because combining it with
the bound $S\left(  P|B\right)  +S\left(  Q|E\right)  \geq-2\log_{2}c\left(
P,Q\right)  $ of Berta \textit{et al}.~leads to the new inequality in
(\ref{eq:crypto-ineq}).)

As a second application of our inequality, we show that it helps to give a
lower bound on the entanglement of formation. Suppose that Alice and Bob share
the state $\rho_{AB}$. Recall that the entanglement of formation $E_{f}\left(
\rho_{AB}\right)  $ and its regularization $E_{f}^{\infty}\left(  \rho
_{AB}\right)  $ are as follows \cite{kavan}:%
\begin{align*}
E_{f}\left(  \rho_{AB}\right)   &  \equiv\inf\sum_{x}p\left(  x\right)
S\left(  \phi_{A}^{x}\right)  ,\\
E_{f}^{\infty}\left(  \rho_{AB}\right)   &  \equiv\lim_{k\rightarrow\infty
}\frac{1}{k}E_{f}(\left(  \rho_{AB}\right)  ^{\otimes k}),
\end{align*}
where the infimum is over all ensembles $\left\{  p\left(  x\right)
,\left\vert \phi_{x}\right\rangle _{AB}\right\}  $ such that $\rho_{AB}%
=\sum_{x}p\left(  x\right)  \left\vert \phi_{x}\right\rangle \left\langle
\phi_{x}\right\vert _{AB}$ and $\phi_{A}^{x} =
\text{Tr}_B\{ \left\vert \phi_{x}\right\rangle  \left\langle \phi_{x}\right\vert _{AB}\}$. A recent result of Carlen and Lieb states that
$E_{f}\left(  \rho_{AB}\right)  \geq-S\left(  A|B\right)  $ \cite{CL12}, from
which it easily follows that $E_{f}^{\infty}\left(  \rho_{AB}\right)
\geq-S\left(  A|B\right)  $, by exploiting the fact that entropies are
additive for tensor-power states. Alice and Bob could each measure the
observable $P$ on their shares of $\rho_{AB}$, or they could each measure the
observable $Q$. Let $p_{e}^{P}$ be the probability that the outcomes of $P$ on
Alice and Bob are different, and let $p_{e}^{Q}$ be the probability that the
outcomes of $Q$ on Alice and Bob are different. Using the Fano inequality, we
have $S(P|B)+S(Q|B)\leq b_{F}$, where 
$b_{F}\equiv h(p_{e}^{P})+p_{e}^{P}\log(d-1)+h(p_{e}^{Q})+p_{e}^{Q}\log(d-1)$.
From the Carlen-Lieb inequality and our inequality in
Theorem~\ref{main-theorem}, we obtain a non-trivial lower bound on the
regularized entanglement of formation:%
\begin{multline}
E_{f}^{\infty}\left(  \rho_{AB}\right)  \geq-2\log_{2}%
c(P,Q)\label{eq:EF-lower-bnd}\\
+\max\left\{  0,D_{A}\left(  \rho_{AB}\right)  -J_{A}(\rho_{AB})\right\}
-b_{F}.
\end{multline}
The above inequality shows that from the measurement incompatibility, the two
error probabilities, and the discord and classical correlations, one can
estimate a lower bound on the regularized entanglement of formation
$E_{f}^{\infty}(\rho_{AB})$. As a simple example, it is clear that the bound
is tight for a Schmidt rank $d$\ maximally entangled state for which
$E_{f}^{\infty}=\log d$, $\max\left\{  0,D_{A}\left(  \rho_{AB}\right)
-J_{A}(\rho_{AB})\right\}  =0$, and by choosing the measurements to be Fourier
conjugate so that $-2\log_{2}c(P,Q)=\log d$ and $b_{F}=0$.

As a further application, we show that our EUR gives
an upper bound on the regularized common randomness distillable by means of
one-way classical communication from Charlie to Bob. In the study of general
resource conversion problems, the notion of distilling common randomness from
a given quantum state plays an important role. As shown in \cite{devetak}, the
distillable common randomness is the regularized version of the classical
correlation. For example, if Charlie and Bob share an arbitrarily large number of copies of $\rho_{CB}$,
then the net amount of correlated classical bits that they can share is given
by $C_{D}^{\rightarrow}(\rho_{CB})=\mathrm{lim}_{k\rightarrow\infty}\frac
{1}{k}J(\left(  \rho_{CB}\right)  ^{\otimes k})$. Using the Koashi-Winter
\cite{kw} equality $C_{D}^{\rightarrow}(\rho_{CB})+E_{f}^{\infty}(\rho
_{AB})=S(\rho_{B})$ for some state $\rho_{AB}$ such that $\rho_{ABC}$ purifies
both $\rho_{AB}$ and $\rho_{CB}$, we can substitute into
(\ref{eq:EF-lower-bnd}) in order to obtain the following upper bound on the
distillable common randomness of $\rho_{CB}$:%
\begin{multline}
C_{D}^{\rightarrow}(\rho_{CB})\leq S(\rho_{B})+2\log_{2}c(P,Q)\\
-\max\left\{  0,D_{A}\left(  \rho_{AB}\right)  -J_{A}(\rho_{AB})\right\}
+b_{F}.
\end{multline}
This is an interesting application of our inequality and the Koashi-Winter
relation, giving us an upper bound on the distillable common randomness across
another partition from the local entropy on $B$, the measurement incompatibility,
the discord and classical correlation on $AB$, and the error probabilities on
$A$ and $B$.

In conclusion, we have proved a new uncertainty relation for conditional
entropic measures which depends on the incompatibility of two quantum
measurements, the conditional entropy, the quantum discord, and the classical
correlations of a state shared between the observed system and
 a quantum memory. This new uncertainty relation tightens that of Berta
\textit{et al}.~whenever the quantum discord of the bipartite state is larger
than its classical correlation. This occurs for several natural examples of
bipartite states including all Werner and isotropic states, and it would be interesting
to characterize the full class of states for which this tightening occurs. Furthermore, using
our new inequality we have given a non-trivial lower bound on the regularized
entanglement of formation for a state shared between Alice and Bob. We have
also shown that our inequality can give an upper bound on the distillable
common randomness of quantum state that Bob shares with a third party.
Our results should have several applications in quantum information
theory, quantum communication, quantum cryptography and precision measurements.

We thank the anonymous referee for helpful comments on our paper and Kavan Modi
for useful discussions.

\appendix

\section{Werner States}

In this appendix, we prove that our lower bound in (8) of the main text is
perfectly tight for the class of higher-dimensional Werner states and
Fourier-conjugate measurements. A higher-dimensional Werner state has the
form:%
\[
\sigma_{AB}=\frac{2\left(  1-\lambda\right)  }{d\left(  d+1\right)  }\Pi
^{+}+\frac{2\lambda}{d\left(  d-1\right)  }\Pi^{-},
\]
where $\Pi^{+}$ is the projector onto the symmetric subspace and $\Pi^{-}$ is
the projector onto the antisymmetric subspace.

\begin{theorem}
For Werner states, the uncertainty sum $S\left(  P|B\right)  +S\left(
Q|B\right)  $ and our lower bound $-2\log c\left(  P,Q\right)  +S\left(
A|B\right)  +D_{A}\left(  \rho_{AB}\right)  -J_{A}\left(  \rho_{AB}\right)
$\ in (8) of the main text coincide, and they are equal to%
\begin{multline*}
-\frac{4\left(  1-\lambda\right)  }{d+1}\log\left(  \frac{2\left(
1-\lambda\right)  }{d+1}\right)  \\
-\frac{2\left(  d-1+2\lambda\right)  }{d+1}\log\left(  \frac{d-1+2\lambda
}{d^{2}-1}\right)  .
\end{multline*}
For these states, the lower bound $-2\log c\left(  P,Q\right)  +S\left(
A|B\right)  $ of Berta \textit{et al}.~is equal to%
\[
-\lambda\log\frac{2\lambda}{d\left(  d-1\right)  }-\left(  1-\lambda\right)
\log\left(  \frac{2\left(  1-\lambda\right)  }{d\left(  d+1\right)  }\right)
.
\]

\end{theorem}

\begin{proof}
For these states, the entropies $S\left(  A\right)  $, $S\left(  B\right)  $,
and $S\left(  A|B\right)  $ are as follows:%
\begin{align*}
S\left(  A\right)   &  =S\left(  B\right)  =\log d,\\
S\left(  A|B\right)   &  =-\lambda\log\frac{2\lambda}{d\left(  d-1\right)  }\\
&  \ \ \ \ \ \ -\left(  1-\lambda\right)  \log\left(  \frac{2\left(
1-\lambda\right)  }{d\left(  d+1\right)  }\right)  -\log d.
\end{align*}
Chitambar proved that the classical correlation $J_{A}\left(  \sigma
_{AB}\right)  $ is equal to \cite{chitambar12}%
\begin{align*}
&  J_{A}\left(  \sigma_{AB}\right)  \\
&  =I(A;B)-D_{A}(\sigma_{AB})\\
&  =\log2d+\lambda\log\frac{\lambda}{d-1}+\left(  1-\lambda\right)
\log\left(  \frac{1-\lambda}{d+1}\right)  \\
&  \ \ \ \ \ \ -\left[  \log\left(  d+1\right)  +\lambda\log\frac{\lambda
}{d-1}+\left(  1-\lambda\right)  \log\left(  \frac{1-\lambda}{d+1}\right)
\right]  \\
&  \ \ \ \ \ \ +\frac{2\left(  1-\lambda\right)  }{d+1}\log\left(
1-\lambda\right)  +\frac{d-1+2\lambda}{d+1}\log\left(  \frac{d-1+2\lambda
}{2\left(  d-1\right)  }\right)  \\
&  =\log\left(  \frac{2d}{d+1}\right)  +\frac{2\left(  1-\lambda\right)
}{d+1}\log\left(  1-\lambda\right)  \\
&  \ \ \ \ \ +\frac{d-1+2\lambda}{d+1}\log\left(  \frac{d-1+2\lambda}{2\left(
d-1\right)  }\right)
\end{align*}
Suppose that the POVMs $P$ and $Q$ correspond to Fourier-conjugate observables
so that $P=\left\{  \left\vert z\right\rangle \left\langle z\right\vert
\right\}  $ and $Q=\left\{  \left\vert \tilde{x}\right\rangle \left\langle
\tilde{x}\right\vert \right\}  $. (In what follows, we just call them $Z=P$
and $X=Q$.) For computing $S\left(  Z|B\right)  $ we need to consider the
following state:%
\begin{equation}
\sigma_{ZB}\equiv\sum_{z}\left(  \left\vert z\right\rangle \left\langle
z\right\vert \otimes I\right)  \sigma_{AB}\left(  \left\vert z\right\rangle
\left\langle z\right\vert \otimes I\right)  .\label{eq:z-dephased-werner}%
\end{equation}
Using the facts that $\Pi^{+}=\left(  I+F\right)  /2$ and $\Pi^{-}=\left(
I-F\right)  /2$ and that%
\[
\sum_{z}\left(  \left\vert z\right\rangle \left\langle z\right\vert \otimes
I\right)  \ F\ \left(  \left\vert z\right\rangle \left\langle z\right\vert
\otimes I\right)  =\sum_{z}\left\vert z\right\rangle \left\langle z\right\vert
\otimes\left\vert z\right\rangle \left\langle z\right\vert ,
\]
it follows that (\ref{eq:z-dephased-werner}) is equal to%
\begin{align}
&  \frac{\left(  1-\lambda\right)  }{d\left(  d+1\right)  }\left(  I+\sum
_{z}\left\vert z\right\rangle \left\langle z\right\vert \otimes\left\vert
z\right\rangle \left\langle z\right\vert \right)  \nonumber\\
&  \ \ \ +\frac{\lambda}{d\left(  d-1\right)  }\left(  I-\sum_{z}\left\vert
z\right\rangle \left\langle z\right\vert \otimes\left\vert z\right\rangle
\left\langle z\right\vert \right)  \nonumber\\
&  =\left(  \frac{\left(  1-\lambda\right)  }{d\left(  d+1\right)  }%
+\frac{\lambda}{d\left(  d-1\right)  }\right)  I\nonumber\\
&  \ \ \ +\left(  \frac{\left(  1-\lambda\right)  }{d\left(  d+1\right)
}-\frac{\lambda}{d\left(  d-1\right)  }\right)  \sum_{z}\left\vert
z\right\rangle \left\langle z\right\vert \otimes\left\vert z\right\rangle
\left\langle z\right\vert \nonumber\\
&  =\frac{d\left(  d-1+2\lambda\right)  }{d^{2}-1}\frac{I}{d^{2}}\nonumber\\
&  \ \ \ +\frac{d-1-2\lambda d}{d^{2}-1}\sum_{z}\frac{1}{d}\left\vert
z\right\rangle \left\langle z\right\vert \otimes\left\vert z\right\rangle
\left\langle z\right\vert \nonumber\\
&  =\sum_{z}\frac{1}{d}\left\vert z\right\rangle \left\langle z\right\vert
\otimes\left(  \frac{d\left(  d-1+2\lambda\right)  }{d^{2}-1}\frac{I}{d}%
+\frac{d-1-2\lambda d}{d^{2}-1}\left\vert z\right\rangle \left\langle
z\right\vert \right)  \label{eq:dephased-Werner}%
\end{align}
From the above form, we deduce that $H\left(  Z\right)  =\log d$, implying
that%
\begin{align*}
S\left(  Z|B\right)   &  =S\left(  B|Z\right)  +H\left(  Z\right)  -S\left(
B\right)  \\
&  =S\left(  B|Z\right)  .
\end{align*}
Consider that each conditional state in (\ref{eq:dephased-Werner}) can be
expressed as%
\begin{align*}
&  \frac{d\left(  d-1+2\lambda\right)  }{d^{2}-1}\frac{I}{d}+\frac
{d-1-2\lambda d}{d^{2}-1}\left\vert z\right\rangle \left\langle z\right\vert
\\
&  =\frac{d-1+2\lambda}{d^{2}-1}\left[  \left(  I-\left\vert z\right\rangle
\left\langle z\right\vert \right)  +\left\vert z\right\rangle \left\langle
z\right\vert \right]  +\frac{d-1-2\lambda d}{d^{2}-1}\left\vert z\right\rangle
\left\langle z\right\vert \\
&  =\frac{d-1+2\lambda}{d^{2}-1}\left(  I-\left\vert z\right\rangle
\left\langle z\right\vert \right)  +\frac{2\left(  1-\lambda\right)  \left(
d-1\right)  }{d^{2}-1}\left\vert z\right\rangle \left\langle z\right\vert .
\end{align*}
For the above state, each probability is independent of the particular value
of $Z$, implying that the entropy $S\left(  B|Z\right)  $ is just the von
Neumann entropy of the following state:%
\[
\frac{d-1+2\lambda}{d^{2}-1}\left(  I-\left\vert z\right\rangle \left\langle
z\right\vert \right)  +\frac{2\left(  1-\lambda\right)  \left(  d-1\right)
}{d^{2}-1}\left\vert z\right\rangle \left\langle z\right\vert .
\]
Thus, we have that%
\begin{align*}
S\left(  Z|B\right)   &  =S\left(  B|Z\right)  \\
&  =-\left(  d-1\right)  \frac{d-1+2\lambda}{d^{2}-1}\log\frac{d-1+2\lambda
}{d^{2}-1}\\
&  \ \ \ -\frac{2\left(  1-\lambda\right)  \left(  d-1\right)  }{d^{2}-1}%
\log\frac{2\left(  1-\lambda\right)  \left(  d-1\right)  }{d^{2}-1}\\
&  =-\frac{d-1+2\lambda}{d+1}\log\frac{d-1+2\lambda}{d^{2}-1}\\
&  \ \ \ -\frac{2\left(  1-\lambda\right)  }{d+1}\log\frac{2\left(
1-\lambda\right)  }{d+1}.
\end{align*}
By a similar line of reasoning, it follows from the high symmetry of the
Werner state that%
\begin{multline*}
S\left(  X|B\right)  =-\frac{d-1+2\lambda}{d+1}\log\frac{d-1+2\lambda}%
{d^{2}-1}\\
-\frac{2\left(  1-\lambda\right)  }{d+1}\log\frac{2\left(  1-\lambda\right)
}{d+1}%
\end{multline*}
Putting everything together, the sum of the uncertainties is equal to%
\begin{multline}
S\left(  Z|B\right)  +S\left(  X|B\right)  =-\frac{2\left(  d-1+2\lambda
\right)  }{d+1}\log\frac{d-1+2\lambda}{d^{2}-1}%
\label{eq:werner-state-uncertainty-sum}\\
-\frac{4\left(  1-\lambda\right)  }{d+1}\log\frac{2\left(  1-\lambda\right)
}{d+1}%
\end{multline}
The lower bound of Berta \textit{et al}.~is as follows:%
\begin{multline*}
-2\log c+S\left(  A|B\right)  =-\lambda\log\frac{2\lambda}{d\left(
d-1\right)  }\\
-\left(  1-\lambda\right)  \log\left(  \frac{2\left(  1-\lambda\right)
}{d\left(  d+1\right)  }\right)  ,
\end{multline*}
while our new lower bound is equal to%
\begin{align*}
&  -2\log c+S\left(  A\right)  -2J_{A}\left(  \rho_{AB}\right)  \\
&  =2\log d-2\log\left(  \frac{2d}{d+1}\right)  -\frac{4\left(  1-\lambda
\right)  }{d+1}\log\left(  1-\lambda\right)  \\
&  \ \ \ -\frac{2\left(  d-1+2\lambda\right)  }{d+1}\log\left(  \frac
{d-1+2\lambda}{2\left(  d-1\right)  }\right)  \\
&  =2\log\left(  \frac{d+1}{2}\right)  -\frac{4\left(  1-\lambda\right)
}{d+1}\log\left(  1-\lambda\right)  \\
&  \ \ \ -\frac{2\left(  d-1+2\lambda\right)  }{d+1}\log\left(  \frac
{d-1+2\lambda}{2\left(  d-1\right)  }\right)
\end{align*}
Since $2=\frac{4\left(  1-\lambda\right)  }{d+1}+\frac{2\left(  d-1+2\lambda
\right)  }{d+1}$, we then have that the above is equal to%
\begin{align*}
&  =\frac{4\left(  1-\lambda\right)  }{d+1}\log\left(  \frac{d+1}{2}\right)
+\frac{2\left(  d-1+2\lambda\right)  }{d+1}\log\left(  \frac{d+1}{2}\right)
\\
&  \ \ \ -\frac{4\left(  1-\lambda\right)  }{d+1}\log\left(  1-\lambda\right)
\\
&  \ \ \ -\frac{2\left(  d-1+2\lambda\right)  }{d+1}\log\left(  \frac
{d-1+2\lambda}{2\left(  d-1\right)  }\right)  \\
&  =-\frac{4\left(  1-\lambda\right)  }{d+1}\log\left(  \frac{2\left(
1-\lambda\right)  }{d+1}\right)  \\
&  \ \ \ -\frac{2\left(  d-1+2\lambda\right)  }{d+1}\log\left(  \frac
{d-1+2\lambda}{d^{2}-1}\right)  .
\end{align*}
This last equality demonstrates that our new lower bound coincides exactly
with the uncertainty sum in (\ref{eq:werner-state-uncertainty-sum}) for all
dimensions $d$ and for all values of $\lambda$.
\end{proof}

Figure~\ref{fig:werner} plots the difference between our new bound and that of
Berta \textit{et al}.~for Werner states as a function of the dimension~$d$ and
the parameter $\lambda$.

\begin{figure}[ptb]
\includegraphics*[width=3.5in,keepaspectratio]{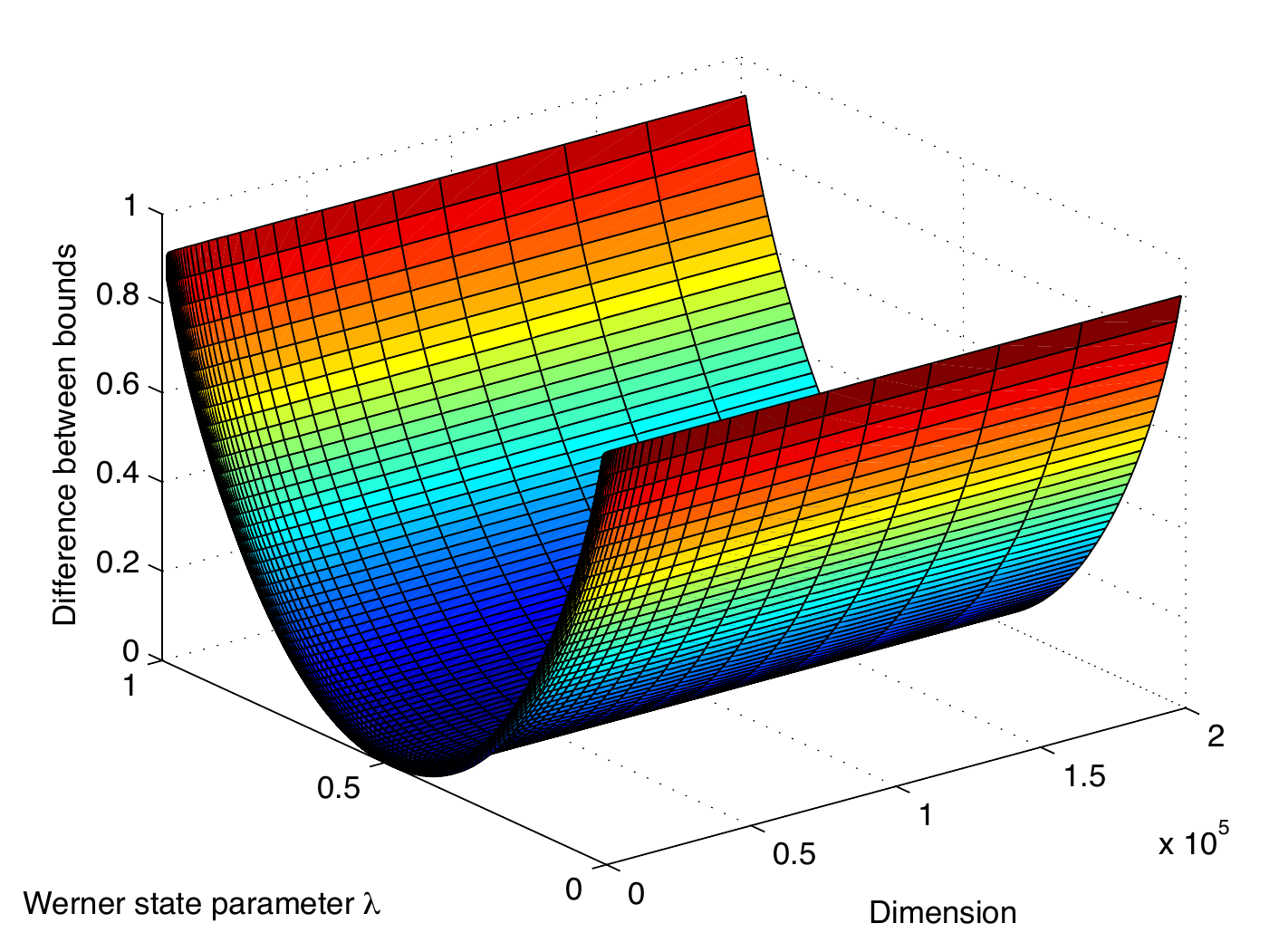}%
\caption{The difference between our new bound in (8) of the main text and that
of Berta \textit{et al}.~for Werner states as a function of the dimension~$d$
and the parameter $\lambda$.}%
\label{fig:werner}%
\end{figure}

\section{Isotropic States}

In this appendix, we prove that our lower bound in (8) of the main text is
perfectly tight for the class of higher-dimensional isotropic states and
Fourier-conjugate measurements. Isotropic states are defined as%
\[
\omega_{AB}=\lambda\Phi_{d}+\frac{1-\lambda}{d^{2}-1}\left(  I-\Phi
_{d}\right)  .
\]

\begin{theorem}
For isotropic states, the uncertainty sum $S\left(  P|B\right)  +S\left(
Q|B\right)  $ and our lower bound $-2\log c\left(  P,Q\right)  +S\left(
A|B\right)  +D_{A}\left(  \rho_{AB}\right)  -J_{A}\left(  \rho_{AB}\right)
$\ in (8) of the main text coincide, and they are equal to%
\[
-\frac{2\left(  d\lambda+1\right)  }{d+1}\log\frac{d\lambda+1}{d+1}%
-\frac{2d\left(  1-\lambda\right)  }{d+1}\log\frac{d\left(  1-\lambda\right)
}{d^{2}-1}%
\]
For these states, the lower bound $-2\log c\left(  P,Q\right)  +S\left(
A|B\right)  $ of Berta \textit{et al}.~is equal to%
\[
-\lambda\log\lambda-\left(  1-\lambda\right)  \log_{2}\left(  \frac{1-\lambda
}{d^{2}-1}\right)  .
\]

\end{theorem}

\begin{proof}
For these states,
\begin{align*}
S\left(  A\right)   &  =S\left(  B\right)  =\log d,\\
S\left(  A|B\right)   &  =-\lambda\log\lambda-\left(  1-\lambda\right)
\log_{2}\left(  \frac{1-\lambda}{d^{2}-1}\right)  -\log d
\end{align*}
and Chitambar proved that \cite{chitambar12}%
\begin{align*}
&  J_{A}\left(  \omega_{AB}\right)  \\
&  =I(A;B)-D_{A}(\omega_{AB})\\
&  =2\log d+\lambda\log\lambda+\left(  1-\lambda\right)  \log_{2}\left(
\frac{1-\lambda}{d^{2}-1}\right)  \\
&  \ \ \ -\lambda\log\lambda-\left(  \frac{1-\lambda}{d+1}\right)  \log\left(
\frac{1-\lambda}{d^{2}-1}\right)  \\
&  \ \ +\frac{d\lambda+1}{d+1}\log\left(  \frac{d\lambda+1}{d\left(
d+1\right)  }\right)  \\
&  =2\log d+\frac{d\left(  1-\lambda\right)  }{d+1}\log\left(  \frac
{1-\lambda}{d^{2}-1}\right)  \\
&  \ \ \ +\frac{d\lambda+1}{d+1}\log\left(  \frac{d\lambda+1}{d\left(
d+1\right)  }\right)
\end{align*}
For computing $S\left(  Z|B\right)  $ we need to consider the following state:%
\begin{equation}
\omega_{ZB}\equiv\sum_{z}\left(  \left\vert z\right\rangle \left\langle
z\right\vert \otimes I\right)  \omega_{AB}\left(  \left\vert z\right\rangle
\left\langle z\right\vert \otimes I\right)  .
\end{equation}
Observing that%
\[
\sum_{z}\left(  \left\vert z\right\rangle \left\langle z\right\vert \otimes
I\right)  \ \Phi_{d}\ \left(  \left\vert z\right\rangle \left\langle
z\right\vert \otimes I\right)  =\sum_{z}\frac{1}{d}\left\vert z\right\rangle
\left\langle z\right\vert \otimes\left\vert z\right\rangle \left\langle
z\right\vert ,
\]
we have that%
\begin{align*}
\omega_{ZB} &  =\lambda\sum_{z}\frac{1}{d}\left\vert z\right\rangle
\left\langle z\right\vert \otimes\left\vert z\right\rangle \left\langle
z\right\vert \\
&  \ \ \ +\frac{1-\lambda}{d^{2}-1}\left(  I-\sum_{z}\frac{1}{d}\left\vert
z\right\rangle \left\langle z\right\vert \otimes\left\vert z\right\rangle
\left\langle z\right\vert \right)  \\
&  =\sum_{z}\frac{1}{d}\left\vert z\right\rangle \left\langle z\right\vert
\otimes\left(  \lambda-\frac{1-\lambda}{d^{2}-1}\right)  \left\vert
z\right\rangle \left\langle z\right\vert \\
&  \ \ \ +\frac{d^{2}\left(  1-\lambda\right)  }{d^{2}-1}\frac{I}{d^{2}}\\
&  =\sum_{z}\frac{1}{d}\left\vert z\right\rangle \left\langle z\right\vert
\otimes\left[  \frac{d^{2}\lambda-1}{d^{2}-1}\left\vert z\right\rangle
\left\langle z\right\vert +\frac{d\left(  1-\lambda\right)  }{d^{2}%
-1}I\right]
\end{align*}
We now focus on rewriting the conditional state above as%
\begin{align*}
&  \frac{d^{2}\lambda-1}{d^{2}-1}\left\vert z\right\rangle \left\langle
z\right\vert +\frac{d\left(  1-\lambda\right)  }{d^{2}-1}I\\
&  =\frac{d^{2}\lambda-1}{d^{2}-1}\left\vert z\right\rangle \left\langle
z\right\vert +\frac{d\left(  1-\lambda\right)  }{d^{2}-1}\left[  \left(
I-\left\vert z\right\rangle \left\langle z\right\vert \right)  +\left\vert
z\right\rangle \left\langle z\right\vert \right]  \\
&  =\frac{d\lambda+1}{d+1}\left\vert z\right\rangle \left\langle z\right\vert
+\frac{d\left(  1-\lambda\right)  }{d^{2}-1}\left(  I-\left\vert
z\right\rangle \left\langle z\right\vert \right)  .
\end{align*}
It follows for the state $\omega_{ZB}$ that $H\left(  Z\right)  =\log d$,
implying that%
\begin{align*}
S\left(  Z|B\right)   &  =S\left(  B|Z\right)  +H\left(  Z\right)  -S\left(
B\right)  \\
&  =S\left(  B|Z\right)
\end{align*}
Observing that the eigenvalues of the above conditional states do not depend
on $z$ in any way, the entropy $S\left(  B|Z\right)  $ is the von Neumann
entropy of the following state:%
\[
\frac{d\lambda+1}{d+1}\left\vert z\right\rangle \left\langle z\right\vert
+\frac{d\left(  1-\lambda\right)  }{d^{2}-1}\left(  I-\left\vert
z\right\rangle \left\langle z\right\vert \right)  .
\]
Calculating directly, we obtain%
\begin{align*}
S\left(  Z|B\right)   &  =S\left(  B|Z\right)  \\
&  =-\frac{d\lambda+1}{d+1}\log\frac{d\lambda+1}{d+1}\\
&  \ \ \ -\left(  d-1\right)  \frac{d\left(  1-\lambda\right)  }{d^{2}-1}%
\log\frac{d\left(  1-\lambda\right)  }{d^{2}-1}\\
&  =-\frac{d\lambda+1}{d+1}\log\frac{d\lambda+1}{d+1}\\
&  \ \ \ -\frac{d\left(  1-\lambda\right)  }{d+1}\log\frac{d\left(
1-\lambda\right)  }{d^{2}-1}.
\end{align*}
By following the same procedure, we can calculate that $S\left(  X|B\right)  $
is as follows:%
\begin{multline*}
S\left(  X|B\right)  =-\frac{d\lambda+1}{d+1}\log\frac{d\lambda+1}{d+1}\\
-\frac{d\left(  1-\lambda\right)  }{d+1}\log\frac{d\left(  1-\lambda\right)
}{d^{2}-1},
\end{multline*}
so that the uncertainty sum is%
\begin{multline}
S\left(  Z|B\right)  +S\left(  X|B\right)  =-\frac{2\left(  d\lambda+1\right)
}{d+1}\log\frac{d\lambda+1}{d+1}\label{eq:isotropic-state-uncertainty-sum}\\
-\frac{2d\left(  1-\lambda\right)  }{d+1}\log\frac{d\left(  1-\lambda\right)
}{d^{2}-1}%
\end{multline}
\begin{figure}[ptb]
\includegraphics*[width=3.5in,keepaspectratio]{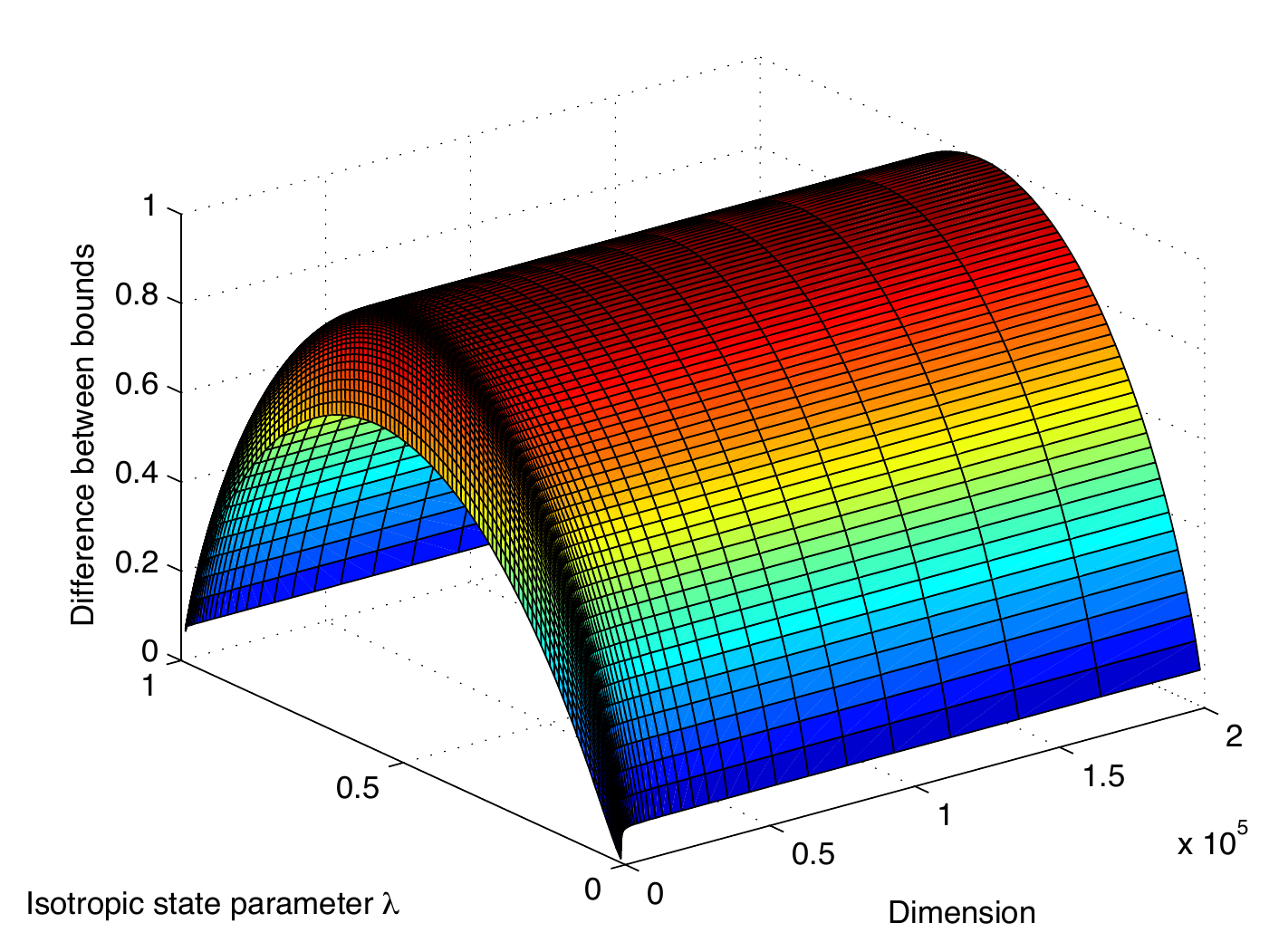}%
\caption{The difference between our new bound in (8) of the main text and that
of Berta \textit{et al}.~for isotropic states as a function of the
dimension~$d$ and the parameter $\lambda$.}%
\label{fig:isotropic}%
\end{figure}The Berta \textit{et al}.~lower bound is then%
\begin{multline*}
-2\log c+S\left(  A|B\right)  =-\lambda\log\lambda\\
-\left(  1-\lambda\right)  \log_{2}\left(  \frac{1-\lambda}{d^{2}-1}\right)  ,
\end{multline*}
while our new lower bound is%
\begin{multline*}
-2\log c+S\left(  A\right)  -2J_{A}\left(  \rho_{AB}\right)  =-2\log d\\
-\frac{2d\left(  1-\lambda\right)  }{d+1}\log\left(  \frac{1-\lambda}{d^{2}%
-1}\right)  -\frac{2\left(  d\lambda+1\right)  }{d+1}\log\left(
\frac{d\lambda+1}{d\left(  d+1\right)  }\right)
\end{multline*}
Since $2=\frac{2d\left(  1-\lambda\right)  }{d+1}+\frac{2\left(
d\lambda+1\right)  }{d+1}$, we can rewrite the above as%
\begin{multline*}
-\frac{2d\left(  1-\lambda\right)  }{d+1}\log d-\frac{2\left(  d\lambda
+1\right)  }{d+1}\log d\\
-\frac{2d\left(  1-\lambda\right)  }{d+1}\log\left(  \frac{1-\lambda}{d^{2}%
-1}\right)  \\
-\frac{2\left(  d\lambda+1\right)  }{d+1}\log\left(  \frac{d\lambda
+1}{d\left(  d+1\right)  }\right)  \\
=-\frac{2\left(  d\lambda+1\right)  }{d+1}\log\frac{d\lambda+1}{d+1}%
-\frac{2d\left(  1-\lambda\right)  }{d+1}\log\frac{d\left(  1-\lambda\right)
}{d^{2}-1}%
\end{multline*}
This last equality demonstrates that our new lower bound coincides exactly
with the uncertainty sum in (\ref{eq:isotropic-state-uncertainty-sum}) for all
dimensions $d$ and for all values of $\lambda$.
\end{proof}

Figure~\ref{fig:isotropic} plots the difference between our new bound and that
of Berta \textit{et al}.~for isotropic states as a function of the
dimension~$d$ and the parameter $\lambda$.

\end{document}